\documentclass[letterpaper, 10 pt, conference]{ieeeconf}  

\IEEEoverridecommandlockouts                              

\usepackage{bm}
\usepackage{nicefrac}
\usepackage{amssymb}
\usepackage{amsmath}

\usepackage{amsthm}
\usepackage{cite}
\usepackage{xcolor}
\usepackage{graphicx}
\usepackage{url}
\usepackage{svg}
\usepackage{balance}
\newtheoremstyle{exampstyle}
  {3pt} 
  {3pt} 
  {\itshape} 
  {} 
  {\bfseries} 
  {.} 
  {.5em} 
  {} 

\theoremstyle{exampstyle} 
\newtheorem{definition}{Definition}

\newtheorem{theorem}{Theorem}
\newtheorem{remark}{Remark}
\newtheorem{assumption}{Assumption}
\newtheorem{proposition}{Proposition}

\newtheorem{example}{Example}
\newtheorem{corollary}{Corollary}

\theoremstyle{plain}

\definecolor{mumred}{RGB}{222,33,77}
\definecolor{mumgreen}{RGB}{0, 140, 0}
\definecolor{mumblue}{RGB}{0, 100, 222}
\definecolor{mumpurple}{RGB}{128, 0, 128}

\newcommand{\ofx}{\left(\bm{x}\right)}

\DeclareMathOperator*{\argmin}{arg\,min}
\newcommand{\xb}{\bm{x}}
\newcommand{\ub}{\bm{u}}

\newcommand{\ddt}{\frac{\mathrm{d}}{\mathrm{d}t}}

\title{\LARGE \bf
	Finding Control Invariant Sets via Lipschitz Constants of Linear Programs
}

\author{Matti Vahs, Shaohang Han and Jana Tumova
	\thanks{This work was partially supported by the Wallenberg AI, Autonomous
		Systems and Software Program (WASP) funded by the Knut and Alice
		Wallenberg Foundation. This research has been carried out as part of the Vinnova Competence Center for Trustworthy Edge Computing Systems and Applications at KTH Royal Institute of Technology.}
	\thanks{The authors are with the Division of Robotics, Perception and Learning, School of Electrical Engineering and Computer Science , KTH Royal Institute of Technology, Stockholm, Sweden and also affiliated with Digital Futures. Mail addresses: {\{\tt\small vahs, shaohang, tumova\}}
		{\tt\small @kth.se}}%
}

\begin{document}
	\maketitle
	\thispagestyle{empty}
	\pagestyle{empty}

	\begin{abstract}
        Control invariant sets play an important role in safety-critical control and find broad application in numerous fields such as obstacle avoidance for mobile robots. However, finding valid control invariant sets of dynamical systems under input limitations is notoriously difficult.
        We present an approach to safely expand an initial set while always guaranteeing that the set is control invariant. Specifically, we define an expansion law for the boundary of a set and check for control invariance using Linear Programs (LPs). To verify control invariance on a continuous domain, we leverage recently proposed Lipschitz constants of LPs to transform the problem of continuous verification into a finite number of LPs. Using concepts from differentiable optimization, we derive the safe expansion law of the control invariant set and show how it can be interpreted as a second invariance problem in the space of possible boundaries. Finally, we show how the obtained set can be used to obtain a minimally invasive safety filter in a Control Barrier Function (CBF) framework. Our work is supported by theoretical results as well as numerical examples.
	\end{abstract}

    \section{INTRODUCTION}
    One of the key concepts in the study of dynamical systems is \emph{forward invariance}, which refers to the property that, if a set contains the system state at some time, it will also contain it in the future \cite{blanchini1999set}. Invariant sets have been studied for decades in the context of Lyapunov stability, constrained control or robustness analysis \cite{blanchini2008set}. For controlled dynamical systems such as, e.g. robots, we say a set is \emph{control invariant} or \emph{viable}, if there exists a feasible control input that renders a set forward invariant.

    Within the field of safety-critical control, numerous methods have been proposed to guarantee control invariance such as Hamilton-Jacobi (HJ) reachability analysis \cite{bansal2017hamilton}, model predictive control \cite{wabersich2021predictive} or Control Barrier Functions (CBFs) \cite{ames2016control}, to name a few. HJ reachability finds the largest control invariant set under input constraints and possible adversaries through dynamic programming. CBFs define a control invariant set as the super level set of a continuously differentiable function that ensures a positive flow into the set for all boundary states. Despite fast progress in the field, finding valid CBFs under actuation constraints is notoriously difficult and still remains an open problem. 

    Early approaches to finding a valid CBF rely on sum-of-squares (SOS) programming \cite{ahmadi2016some, clark2021verification, wang2023safety, xu2017correctness} which can be readily used when the system dynamics are represented as polynomials. A recent study \cite{choi2021robust} has explored the relationship between HJ reachability and CBFs, enabling the synthesis of CBFs for finite horizon problems directly from HJ reachability calculations. Similar to standard HJ reachability problems, the authors discretize the state space and solve the HJ Isaacs equation via dynamic programming.

    Backup CBFs\cite{gurriet2018online, chen2021backup}, which are closely related to model predictive safety filters \cite{wabersich2022predictive} , define an implicit control invariant set by guaranteeing that there always exists a backup trajectory that will lead to a smaller control invariant set. However, since the maximum control invariant set is only defined implicitly, its exact size remains unknown before deployment. Furthermore, the size of the implicit control invariant set heavily depends on the backup strategy which typically includes domain specific knowledge.
    \begin{figure}[t!]
        \centering
        \includegraphics[scale=0.46]{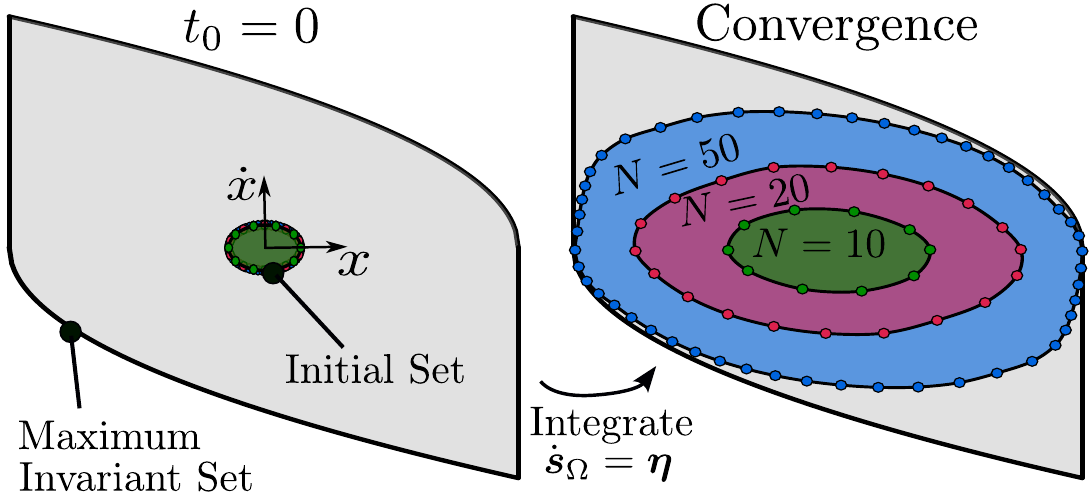}
        \caption{Illustration of the proposed invariant set expansion applied to the double integrator. Starting from a small invariant set, we apply a virtual expansion control input to each control point and ensure that there exists a control input keeping the state inside the set for all states on the boundary. Here, we vary the number of control points on the boundary to illustrate the conservatism.}
        \label{fig:FirstPage}
    \end{figure}

    Alternatively to finding exact control invariant sets, many recent works leverage neural networks (NNs) to approximate CBFs which are often referred to as neural CBFs \cite{so2024train, dawson2022safe}. Although neural networks can represent arbitrary complex functions, it is important to stress that neural CBFs need to be verified since NNs are trained on a finite set of data points. There are a variety of verification tools such as branch and bound schemes \cite{bunel2020branch} or Satisfiability Modulo Theory (SMT) \cite{peruffo2021automated}. In \cite{wang2024simultaneous}, a verification-in-the-loop scheme is proposed to simultaneously synthesize and verify a neural CBF.

    Closest related to our work are approximated regions of attraction (RoA) using flow CBFs \cite{ubellacker2024approximating}. Given a small RoA, the authors expand a polygon while ensuring that a backup CBF exists for all vertices of the polygon that drives the state back to the smaller RoA. The authors show that they approach the exact RoA as the number of vertices approaches infinity. However, for any finite number, it can happen that a state on the boundary leaves the set as the authors only verify the backup strategy at discrete vertex positions.
    In this work, we show how to define CBFs in the space of possible boundaries which allow for safe expansion of control invariant sets. Similar ideas have been proposed in our previous work, where CBFs have been applied to alternative domains to robot state spaces such as belief spaces \cite{10310096, vahs2024risk} or trajectory spaces \cite{vahs2024forward}.

    Specifically, we present an approach for finding valid control invariant sets on a continuous domain. Instead of using the zero level set of a continuously differentiable function, we directly parameterize the boundary of an invariant set and ensure that for all points on the boundary, there exists a positive flow into the set. This greatly simplifies the process of checking if a set is invariant since we can check the flow condition at discrete points on the boundary through a Linear Program (LP). To guarantee the flow condition on a continuous domain, i.e. all states on the boundary, we leverage recently proposed Lipschitz constants of LPs \cite{canovas2022projection}. Finally, we propose a safe expansion law that increases the size of a set while ensuring that the set is always control invariant, see Fig.~\ref{fig:FirstPage}. 



 



        
    \section{Preliminaries}
    \subsection{Control Barrier Functions}
    Consider a dynamical system in control affine form
\begin{align}
    \dot{\bm{x}} = \bm{F}(\xb, \ub) =\bm{f}\left(\bm{x}\right) + \bm{g}\left(\bm{x}\right) \bm{u}\label{eq:deterministic},
\end{align}
with state $\bm{x} \in \mathcal{X} \subseteq \mathbb{R}^n$ and control input $\bm{u} \in \mathcal{U} \subseteq \mathbb{R}^m$. 
A safe set $\mathcal{C}$ is constructed as the superlevel set of a continuously differentiable function $h: \mathcal{X} \rightarrow \mathbb{R}$ such that
\begin{equation}
\begin{aligned}
  \label{eq:safeset_state}
    \mathcal{C} &= \left\{\bm{x} \in \mathcal{X} \mid h\left(\bm{x}\right) \geq 0\right\},\\
    \partial \mathcal{C} &= \left\{\bm{x} \in \mathcal{X} \mid h\left(\bm{x}\right) = 0\right\}.
 \end{aligned}
\end{equation} 
\begin{definition}
A set $\mathcal{C}$ is control invariant with respect to the system \eqref{eq:deterministic} if for every initial condition $\bm{x}(t_0) \in \mathcal{C}$
there exists a control input $\ub(t) \in \mathcal{U}$ such  that $\bm{x}(t) \in \mathcal{C}, \forall t \geq t_0$.
\end{definition}
A prominent approach to render a safe set control invariant is to use CBFs.
\begin{definition}
\label{def:CBF}
Given a set $\mathcal{C}$, defined by Eq. \eqref{eq:safeset_state}, $h$ serves as a zeroing CBF for the system \eqref{eq:deterministic} if  $\forall \bm{x}$ satisfying $h\left(\bm{x}\right) \geq 0, \exists \bm{u} \in \mathcal{U}$ such that
\begin{align}
    \frac{\partial h}{\partial \bm{x}} \left(\bm{f}\left(\bm{x}\right) + \bm{g}\left(\bm{x}\right) \bm{u}\right) \geq - \alpha(h\left(\bm{x}\right)).\label{eq:cbf}
\end{align}
where $\alpha$ is a class-$\kappa$ function, i.e. continuous and strictly increasing with $\alpha(0)=0$.
\end{definition}
In case a valid CBF exists, it follows that a control input $\ub$ satisfying Eq.~\eqref{eq:cbf} renders $\mathcal{C}$ control invariant \cite{ames2016control}.

\subsection{Box Constrained Linear Programs}
\label{sec:LP}
Consider a box constrained linear program (LP) of the form
\begin{equation}
    \begin{aligned}
        \vartheta = \underset{\bm{w} \in \mathbb{R}^{\ell}}{\mathrm{min}}  \quad & \bm{c}^T \bm{w}  \\
        \textrm{s.t.~~} \quad & w_{i, \mathrm{min}} \leq w_i \leq w_{i, \mathrm{max}},~ i=1,\dots, \ell
    \end{aligned}
    \label{eq:LPGeneral}
\end{equation}
where $\bm{w}$ is the decision variable, $\bm{c} \in \mathbb{R}^{\ell}$ is the cost vector and $w_{i, \mathrm{min}}, w_{i, \mathrm{max}}$ are the $i$th minimum and maximum values $\bm{w}$, respectively.  We will use the shorthand notation $\vartheta(\bm{c})$ or just $\vartheta$ to refer to the solution of an LP parameterized by $\bm{c}$. For this specific case of an LP, we can obtain the closed form solution
\begin{align}
    w_i^* = \begin{cases}
       w_{i, \mathrm{min}} & \mathrm{if}~c_i > 0\\
       w_{i, \mathrm{max}} & \mathrm{if}~c_i < 0\\
       \in [w_{i, \mathrm{min}}, w_{i, \mathrm{max}}] & \mathrm{if}~c_i = 0
    \end{cases}\label{eq:closedform}
\end{align}
since the optimal solution of an LP always lies on the boundary of the feasible set \cite{dantzig1955generalized}. In this work, we consider $\bm{c}$-perturbations of the LP~\eqref{eq:LPGeneral}, meaning that we are interested in the sensitivity of the optimal value $\vartheta$ with respect to changes in $\bm{c}$. From the closed form expression in Eq.~\eqref{eq:closedform}, we conclude that for any $c_i \neq 0$
\begin{align*}
    \frac{\partial \vartheta}{\partial c_i} = \frac{\partial}{\partial c_i} \bm{c}^T \bm{w}^* = w_i^*.
\end{align*}

It has been shown in \cite{gisbert2019lipschitz} that the optimal value $\vartheta$ is continuous with respect to changes in $\bm{c}$. Moreover, global Lipschitz arguments can be found that bound the change in the optimal value, i.e.
\begin{align}
    \mathcal{L}_{LP} = \underset{\bm{c}, \tilde{\bm{c}}}{\mathrm{sup}} \frac{|\vartheta(\bm{c}) - \vartheta(\tilde{\bm{c}})|}{\lVert \bm{c} - \tilde{\bm{c}}\rVert}\label{eq:LipschitzDef}.
\end{align}
For the following theorem, we assume the inequality constraints in Eq.~\eqref{eq:LPGeneral} to be given in the form $\bm{A} \bm{w} \leq \bm{b}.$ Let us introduce the set
\begin{align}
    \mathcal{F}(\bm{A}, \bm{b}) := \left\{\bm{w} \in \mathbb{R}^{\ell}:\bm{A} \bm{w} \leq \bm{b}\right\}\label{eq:feasibleset},
\end{align}
which describes the feasible set of an inequality constrained LP.
\begin{theorem}{[Thm. 4 in \cite{canovas2022projection}]}
    Consider an LP parameterized by $\bm{A}, \bm{b}$ and $\bm{c}$, and a Lipschitz constant as defined in Eq.~\eqref{eq:LipschitzDef}, then
    \begin{align*}
        \mathcal{L}_{LP} \leq \underset{\bm{w} \in \mathcal{E}}{\mathrm{max}} \lVert \bm{w} \rVert_2,
    \end{align*}\label{thm:Lipschitz}
    where
    \begin{align*}
        \mathcal{E} = \mathrm{extr}\left(\mathcal{F}(\bm{A}, \bm{b}) \cap \mathrm{span}\left\{\bm{a}_i, i=1,\dots,2\ell\right\}\right)
    \end{align*}
    with $\mathcal{F}$ as defined in Eq.~\eqref{eq:feasibleset} and $\bm{a}_i$ being the $i$th row of $\bm{A}$.
\end{theorem}
    Theorem~\ref{thm:Lipschitz} can be further simplified by considering the structure of the box constraints in \eqref{eq:LPGeneral} which we summarize in the following proposition.
    \begin{proposition}
        Consider a box constrained LP as defined in Eq.~\eqref{eq:LPGeneral}, then
    \begin{align*}
        \mathcal{L}_{LP} \leq \Bigg\lVert \begin{bmatrix}
            \mathrm{max}(|w_{1, \mathrm{min}}|, |w_{1, \mathrm{max}}|)\\
            \vdots\\
            \mathrm{max}(|w_{\ell, \mathrm{min}}|, |w_{\ell, \mathrm{max}}|)
        \end{bmatrix} \Bigg\rVert_2
    \end{align*}
    denotes the global Lipschitz constant of \eqref{eq:LPGeneral} under $\bm{c}$-perturbations.\label{lem:Lipschitz}
    \end{proposition}
    \begin{proof}
    We observe that the constraints of Eq.~\eqref{eq:LPGeneral} can be written as
    \begin{align*}
        \bm{A} = \begin{bmatrix}
                \bm{I}_{\ell}\\
                -\bm{I}_{\ell}
            \end{bmatrix} \in \mathbb{R}^{2\ell \times \ell}, \hspace{0.5cm} \bm{b} = \begin{bmatrix}
                \bm{w}_{\mathrm{min}}\\
                -\bm{w}_{\mathrm{max}}
            \end{bmatrix} \in \mathbb{R}^{\ell},
    \end{align*}
    where $\bm{I}_{\ell}$ denotes the identity matrix of size $\ell$. Thus, it since $\mathrm{rank}(\bm{A}) = \ell$, it follows that 
    \begin{align*}
        \mathrm{span}\left\{\bm{a}_i, i=1,\dots,2\ell\right\} = \mathbb{R}^{\ell},
    \end{align*}
    which implies that the set $\mathcal{E}$ reduces to the extreme points of the feasible set. Solving for the maximum norm on the set $\mathcal{E}$ yields the Lipschitz constant in Lemma~\ref{lem:Lipschitz}.
    \end{proof}
    \subsection{Catmull-Rom Curves}
    \label{sec:catmullrom}
    Catmull-Rom curves are a family of piecewise cubic interpolating splines with smooth interpolation between control points $\bm{p}_i \in \mathbb{R}^n$. 
    A single curve segment between points $ \bm{p}_{i}$ and $ \bm{p}_{i+1} $ is defined by four consecutive control points: $ \bm{p}_{i-1} $, $ \bm{p}_i $, $ \bm{p}_{i+1} $, and $ \bm{p}_{i+2} $, where $i + k = (i + k \mod N)$ if $i+k \notin \{0, \dots, N-1\}$. An example with $N=4$ control points is shown in Fig.~\ref{fig:CatmullRom}. 
    
    
    A point $\bm{p}$ along the curve segment between $\bm{p}_i$ and $\bm{p}_{i+1}$ is described by a phasing variable $
    \tau_i \leq \tau \leq \tau_{i+1}$ such that
    \begin{align*}
        \bm{p}(\tau_i) &= \bm{p}_i, &&\bm{p}(\tau_{i+1}) = \bm{p}_{i+1}.
    \end{align*}
    In uniform parameterizations, the phasing variable is set to $1 \leq \tau \leq 2$ which, however, can lead to cusps and self-intersections \cite{yuksel2009parameterization}. To prevent this, \cite{yuksel2009parameterization} proposed non-uniform parameterizations in which the phasing variable is set to
    \begin{align*}
        \lVert\bm{p}_{i} - \bm{p}_{i-1}\rVert_2^{\beta} \leq \tau \leq \lVert\bm{p}_{i} - \bm{p}_{i-1}\rVert_2^{\beta} + \lVert\bm{p}_{i+1} - \bm{p}_{i}\rVert_2^{\beta} ,
    \end{align*}
    with $\beta=0.5$. The Catmull-Rom curve provides $C^1$ continuity and Lipschitz arguments can be found easily due to closed-form expressions of $\bm{p}(\tau)$. 
    Given a set of $N$ control points, multiple Catmull-Rom curve segments can be cascaded to obtain a closed curve as shown in Fig.~\ref{fig:CatmullRom}.
    
    Further, one can obtain bounding boxes for each curve segment with a height of $h = \tfrac{1}{4}\lVert\bm{p}_{i+1} - \bm{p}_i\rVert_2$ as shown in Fig.~\ref{fig:CatmullRom}. For a detailed description of Catmull-Rom curves, the reader is referred to \cite{yuksel2009parameterization}.
    \begin{figure}[t!]
        \centering
        \includegraphics[scale=2.4]{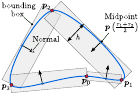}
        \caption{Illustration of a closed Catmull-Rom curve consisting of four segments. The control and mid points are shown by red and blue dots, respectively. The grey rectangles show the bounding box of each curve segment.}
        \label{fig:CatmullRom}
    \end{figure}
    \section{Problem Statement}
    We are concerned with finding a control invariant set for nonlinear control-affine systems as defined in Eq.~\eqref{eq:deterministic} with a feasible input set
    \begin{align}
        \mathcal{U} = \left\{\ub \in \mathbb{R}^m \mid \ub_{\mathrm{min}} \leq \ub \leq \ub_{\mathrm{max}}\right\}.\label{eq:FeasibleInputs}
    \end{align}
    We assume we are given a notion of a safe set $\mathcal{S} \subset \mathcal{X}$
    which contains all possible states that are considered safe, e.g. all collision free states of a robot. We seek to find the largest control invariant set $\Omega \subset \mathcal{X}$ that is fully contained in $\mathcal{S}$. 
    \begin{assumption}
        We represent the boundary of the control invariant set $\Omega$ as a set of control points, i.e. $\bm{x}_i \in \mathcal{X},~i = 1,\dots,N$ and assume the existence of a $C^1$ manifold connecting all discrete points.
    \end{assumption}
    One possible choice of such continuously differentiable manifold in two dimensions are the Catmull-Rom curves introduced in Sec.~\ref{sec:catmullrom}.
    Given $N$, we formalize the task as an optimization problem, reading
    \begin{equation}
        \begin{aligned}
            \underset{\Omega}{\mathrm{max}}  \quad & \mathcal{V}\left(\Omega\right)\\
            \textrm{s.t.~~} \quad &
             \Omega \subseteq \mathcal{S},\\
            ~ \quad &\exists \bm{u} \in \mathcal{U}: \bm{n}(\xb)^T \bm{F}(\xb, \ub) \geq 0,~ \forall \xb \in \partial \Omega,
        \end{aligned}
        \label{eq:Problem}
    \end{equation}
    where $\mathcal{V}(\cdot)$ denotes the volume and $\bm{n}(\xb)$ is the inwards pointing normal at state $\xb$ on the boundary $\partial \Omega$. 
    By enforcing that there exist a control input such that the dynamics flow has a normal component to the boundary, we ensure that $\Omega$ is a control invariant set.
    \begin{example}
        (Running Example) Consider a one-dimensional double integrator $\ddot{p} = u$ with state dynamics
        \begin{align*}
        \dot{\bm{x}} = \begin{bmatrix}
            \dot{p}\\
            \ddot{p}
        \end{bmatrix} = \begin{bmatrix}
            \dot{p}\\
            0
        \end{bmatrix}  + \begin{bmatrix}
            0\\
            1
        \end{bmatrix} u.
        \end{align*}
        The set of physical safe states is defined as $\mathcal{S} = \left\{\bm{x} \in \mathbb{R}^2 \mid |p| \leq 1\right\}$ and the set of feasible control inputs is given as $\mathcal{U} = \left\{u \in \mathbb{R}\mid -1\leq u \leq 1\right\}$. We seek to find the largest control invariant set $\Omega$ contained in $\mathcal{S}$.
    \end{example}

    \section{Safe Invariant Set Expansion}
    In the following, we present our solution to finding a control invariant set. Specifically, our solution approach follows three steps: 1) We model the boundary of the set $\Omega$ as a dynamical system and introduce its continuous time dynamics which we use to 2) find a safe expansion law enforcing control invariance at discrete control point positions and 3) use Lipschitz arguments of linear programs to enforce control invariance on a continuous domain.

    \subsection{Invariant Set Dynamics}
    Our considered sets $\Omega$ are parameterized by $N$ control points $\bm{x}_1,\dots,\bm{x}_N$ and a $C^1$ manifold connecting them where we do not have any requirements on convexity of the set. Thus, given the smooth connecting manifold, the boundary is uniquely described by the position of all control points which allows us to define the set state
    \begin{align*}
        \bm{s}_{\Omega} = \begin{bmatrix}
            \bm{x}_1 & \bm{x}_2 & \dots & \bm{x}_N
        \end{bmatrix}^T \in \mathbb{R}^{n \cdot N}.
    \end{align*}
    Next, we model the set state as a controlled dynamical system that evolves in continuous time. Specifically, we model the dynamics of $\bm{s}_{\Omega}$ as a single integrator
    \begin{align}
        \dot{\bm{s}}_{\Omega}(t) = \bm{\eta}(t)\label{eq:PolytopeSystem},
    \end{align}
    where $\bm{\eta} \in \mathbb{R}^{n\cdot N}$ is a virtual control input that directly affects the velocity of the control points $\bm{x}_i$. With the definition of this dynamical system of the set $\Omega$, one can reformulate the problem in Eq. \eqref{eq:Problem} as finding virtual control inputs $\bm{\eta}$ that increase the volume of the set while ensuring that it stays control invariant. Thus, we introduce a reference control input $\bm{\eta}^{\mathrm{ref}}$ that tries to maximize the volume $\mathcal{V}(\Omega)$ regardless of the constraints in Eq.~\eqref{eq:Problem}. Afterwards, we propose a CBF-based safety filter that corrects $\bm{\eta}^{\mathrm{ref}}$ minimally to ensure control invariance.

    For the reference control input, one possible choice is to directly use the gradient $\nicefrac{\partial \mathcal{V}(\Omega)}{\partial \bm{s}_{\Omega}}$ to maximize the invariant set's volume. However, a closed form expression of the volume is often only available for simple set representations. Generally, we can apply any hand-crafted reference control input that increases the volume of $\Omega$ over time.
    \begin{example}
        (Example Cont.) We parameterize the boundary of the control invariant set $\Omega$ by $N$ control points $\bm{x}_i \in \mathbb{R}^2$ which we stack in $\bm{s}_{\Omega}$ in counter clockwise order as it can be seen in 
        Fig.~\ref{fig:CatmullRom} for $N=4$. We use the Catmull-Rom curve as a smooth function $\bm{x}(\tau)$ connecting all control points, forming a closed curve. We adapt the expansion control law in \cite{ubellacker2024approximating} which is derived for 2D polygons. The reference control input for control point $\bm{x}_i$ is obtained through
    \begin{align*}
        \bm{\eta}^{\mathrm{ref}}_i = k_n \bm{q}_i + k_c (\bm{x}_c^i - \mathrm{proj}_{l_i}(\bm{x}_i)),
    \end{align*}
    for some controller gains $k_n, k_c \geq 0$. Figure~\ref{fig:ExpansionControl} illustrates the components of the expansion control inputs. The component in the direction of $\bm{q}_i$ tries to increase the distance along the normal direction while the other component tries to keep a balanced distribution of control points along the boundary by centering the control point on the connecting line between its neighbor points.
    \end{example}
    \begin{figure}[t!]
        \centering
        \includegraphics[scale=1.]{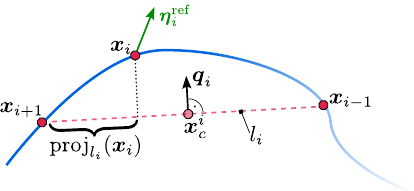}
        \caption{Illustration of the reference virtual control input $\bm{\eta}^{\mathrm{ref}}$ that balances expansion and distribution of control points along the boundary.}
        \label{fig:ExpansionControl}
        \vspace{-0.5cm}
    \end{figure}

    \subsection{Safe Expansion}
    Having defined a reference expansion controller, we seek to correct the virtual control input $\bm{\eta}^{\mathrm{ref}}$ in order to satisfy the constraints in Eq. \eqref{eq:Problem}. To that end, we define a CBF candidate for the dynamical system in Eq. \eqref{eq:PolytopeSystem} which we discuss in detail subsequently.

    By Nagumo's Theroem \cite{blanchini1999set}, the set $\Omega$ is a control invariant set if we can ensure 
    \begin{align}
        \exists \ub \in \mathcal{U}:\bm{n}(\xb)^T \bm{F}(\xb, \ub) \geq 0, \forall \xb \in \partial \Omega,
        \label{eq:existence}
    \end{align}
    which states that there exists a flow into the set for all states on the boundary. More importantly, we can transform this binary condition into a continuously valued function
    \begin{equation}
        \begin{aligned}
            b^*(\bar{\xb}) = \underset{\bm{u} \in \mathbb{R}^m}{\mathrm{sup}}  \quad & \bm{n}(\bar{\xb})^T \left(\bm{f}(\bar{\xb}) + \bm{g}(\bar{\xb}) \ub \right)  \\
            \textrm{s.t.~~} \quad & \ub \in \mathcal{U},
        \end{aligned}
        \label{eq:LPfeasibility}
    \end{equation}
    for one specific state $\bar{\xb} \in \partial \Omega$. Note that, if $\mathcal{U}$ is defined by linear constraints as in Eq. \eqref{eq:FeasibleInputs}, the problem can be formulated as a linear program (LP). Consequently, it holds that
    \begin{align*}
        b^*(\bar{\xb}) \geq 0 \Leftrightarrow \exists \ub \in \mathcal{U}:\bm{n}(\bar{\xb})^T \bm{F}(\bar{\xb}, \ub) \geq 0,
    \end{align*}
    implying that we can define a \emph{safe set in the space of possible sets $\Omega$} as the collection of zero levelsets of $b^*$, i.e.
    \begin{align*}
        \mathcal{C}_{\Omega} = \left\{\Omega \mid b^*(\xb) \geq 0, \forall \xb \in \partial \Omega\right\}.
    \end{align*}
    Thus, we can view the safe expansion of the control invariant set as yet another invariance problem for $\Omega$: \emph{if we can ensure that $\mathcal{C}_{\Omega}$ is control invariant with respect to the system~\eqref{eq:PolytopeSystem}, the set $\Omega$ is always control invariant with respect to the system~\eqref{eq:deterministic}.}

    Next, we present a CBF candidate function that is used to render $\mathcal{C}_{\Omega}$ control invariant. We ensure that the CBF condition
    \begin{align}
        \dot{b}^*(\bm{x}) \geq -\alpha\left(b^*(\xb)\right),~\forall \xb \in \partial \Omega\label{eq:contCBFConstraint},
    \end{align}
    holds for all states $\xb$ on the boundary. If these conditions hold at all times, it follows that $\mathcal{C}_{\Omega}$ is control invariant \cite{ames2016control}. To obtain the time derivative of the optimal value $b^*$, we can rewrite the continuous state on the boundary as a function of the set state $\bm{s}_{\Omega}$ and apply the chain rule
    \begin{align*}
        \dot{b}^*(\bm{x}) &= \ddt b^*\left(\bm{x}(\bm{s}_{\Omega})\right) = \frac{\partial b^*}{\partial \bm{s}_{\Omega}} \dot{\bm{s}}_{\Omega}\\
        \frac{\partial b^*}{\partial \bm{s}_{\Omega}} &= \frac{\partial}{\partial \bm{s}_{\Omega}} \left( \bm{n}^T \left(\bm{f}(\xb) + \bm{g}(\xb) \ub^* \right)\right)\\
        &=\frac{\partial \bm{n}^T}{\partial \bm{s}_{\Omega}} \left(\bm{f}(\xb) + \bm{g}(\xb) \ub^* \right) + \bm{n}^T \left(\frac{\partial \bm{f}}{\partial \bm{s}_{\Omega}} + \frac{\partial \bm{g}}{\partial \bm{s}_{\Omega}} \ub^* \right),
    \end{align*}
    where $\bm{u}^*$ is the optimal solution of the feasibility problem in Eq.~\eqref{eq:LPfeasibility}. Here, we use the fact that $\nicefrac{\partial \ub^*}{\partial \bm{s}_{\Omega}} = \bm{0}$ almost everywhere which follows directly from the discussion in Sec.~\ref{sec:LP} as Eq.~\eqref{eq:LPfeasibility} is a box constrained LP.
    
    Finally, we can obtain a safe virtual control input $\bm{\eta}$ that minimally deviates from the reference control law by solving
    \begin{equation}
        \begin{aligned}
            \bm{\eta}(\bm{s}_{\Omega}) = \argmin_{\bm{\eta} \in \mathbb{R}^{n\cdot N}} \quad & \left(\bm{\eta} - \bm{\eta}^{\mathrm{ref}}\right)^T {\bm{Q}} \left(\bm{\eta} - \bm{\eta}^{\mathrm{ref}}\right)\\
            \textrm{s.t.~~} \quad & \frac{\partial b^*}{\partial \bm{s}_{\Omega}} \bm{\eta} \geq - \alpha(b^*),~\forall \xb \in \partial \Omega.
        \end{aligned}
        \label{eq:QP}
    \end{equation}
    
    However, since the boundary of $\Omega$ is defined on a continuous domain, Eq.~\eqref{eq:contCBFConstraint} would result in infinitely many CBF constraints which prevents us from obtaining a safe expansion input $\bm{\eta}$ in practice. Therefore, we leverage Lipschitz arguments of LPs which we discuss in the following section.

    \subsection{From Infinite to Finite Constraints}
    When further inspecting the feasibility problem in Eq.~\eqref{eq:LPfeasibility} with a feasible control input set as defined in Eq.~\eqref{eq:FeasibleInputs}, we can rewrite the linear program in standard form of a box constrained LP
    \begin{equation}
        \begin{aligned}
            b^*(\xb) = \bm{n}(\xb)^T \bm{f}(\xb) - \underset{\bm{u} \in \mathbb{R}^m}{\mathrm{min}}  \quad & \bm{c}(\xb)^T \ub   \\
            \textrm{s.t.~~} \quad & \ub \in \mathcal{U},
        \end{aligned}
        \label{eq:LPrewritten}
    \end{equation}
    where $\bm{c}(\xb)^T = \bm{n}(\xb)^T \bm{g}(\xb)$. Hence, given Lipschitz constants $\mathcal{L}_{nf}$ of $\bm{n}^T \bm{f}$ and $\mathcal{L}_{ng}$ of $\bm{n}^T\bm{g}$, we can bound the maximum change of the optimal value of the feasibility problem as
    \begin{align}
        |b^*(\xb) - b^*(\tilde{\xb})| &\leq \mathcal{L}_b \lVert \bm{x} - \tilde{\xb} \rVert_2,~\hspace{0.5cm}\text{with}\nonumber\\
        \mathcal{L}_b &= \mathcal{L}_{nf} + \mathcal{L}_{LP} \cdot \mathcal{L}_{ng}\label{eq:LipschitzConstants}
    \end{align}
    where we use the Lipschitz constant of a box constrained LP derived in Prop.~\ref{lem:Lipschitz}. Consequently, instead of checking the feasibility problem on a continuous domain, we can evaluate a single LP as in Eq.~\eqref{eq:LPrewritten} at a state $\bar{\xb}$ and ensure that $b^*(\xb)\geq0$ in a neighborhood around $\bar{\xb}$.
    \begin{example}
        (Example Cont.) We consider the Lipschitz constants for one curve segment defined by four control points and $\tau_1 \leq \tau\leq \tau_2$. To obtain $\mathcal{L}_b$ in Eq.~\eqref{eq:LipschitzConstants}, we need the Lipschitz constant of the normal $\bm{n}(\bm{x}(\tau))$ along a Catmull-Rom curve which is equivalent to the Lipschitz constant of the tangent vector $\bm{t} = \bm{x}'(\tau)$ since $\bm{n}\perp \bm{t}$. Thus, since $\bm{x}'(\tau)$ is a second order polynomial the Lipschitz constant of $\bm{n}$ is
        \begin{align*}
            \mathcal{L}_n = \mathrm{max}(|\bm{x}''(\tau_1)|, |\bm{x}''(\tau_2)|),
        \end{align*}
        as $\bm{x}''(\tau)$ is linear in $\tau$. It follows directly that $\mathcal{L}_{ng} = \mathcal{L}_{n}$ because $\bm{g}(\xb)= [0~1]^T$. For $\bm{n}(\xb)^T\bm{f}(\xb) = n_1(\tau) \cdot x_2(\tau)$, we obtain an upper bound on the Lipschitz constant as
        \begin{align*}
            \mathcal{L}_{nf} \leq M_x \cdot \mathcal{L}_n + M_n \cdot \mathcal{L}_{x_2},
        \end{align*}
        with $\mathcal{L}_{x_2}$ being the Lipschitz constant of a third order polynomial, $M_x = \mathrm{sup}_{\tau} |\bm{x}_2(\tau)|$ can be obtained from the bounding box shown in Sec.~\ref{sec:catmullrom} and and $M_n = \mathrm{sup}_{\tau} \lVert \bm{n} \rVert_2= 1$ as the normal is always normalized. Finally, by Eq.~\eqref{eq:LipschitzConstants}, we obtain the Lipschitz constant of $b^*(\xb)$ as $\mathcal{L}_b(\xb(\tau)) = \mathcal{L}_{nf} + \mathcal{L}_{LP}\cdot \mathcal{L}_{ng}$.
    \end{example}
    \begin{theorem}
        The set $\Omega$ with boundary $\partial \Omega$ defined by $\bm{s}_{\Omega}$ in Eq.~\eqref{eq:safeset_state} and a $C^1$ manifold passing through all control points is a control invariant set if 
        \begin{align*}
            \partial \Omega &\subset \bigcup_{i=1}^N \mathcal{B}_{r_i}(\bm{x}_i), ~\text{and}\\
            b^*(\bm{x}_i) &\geq \mathcal{L}_b \cdot r_i\geq 0, \hspace{0.5cm} \forall i=1,\dots, N
        \end{align*}
        with a Lipschitz constant defined in Eq.~\eqref{eq:LipschitzConstants} and $\mathcal{B}_r(\xb)$ denoting a ball of radius $r\geq 0$ located at $\bm{x}$.\label{thm:main}
    \end{theorem}
    \begin{proof}
        Consider a single ball $\mathcal{B}_{r_i}(\bm{x}_i)$. If the conditions in Thm.~\ref{thm:main} hold, we know that $\forall \bm{x} \in \mathcal{B}_{r_i}(\bm{x}_i)$
        \begin{align*}
            |b^*(\xb) - b^*(\bm{x}_i)| &\leq \mathcal{L}_b \lVert \bm{x} - \bm{x}_i\rVert_2 \leq \mathcal{L}_b\cdot r_i\\
            \Rightarrow b^*(\xb) &\geq b^*(\bm{x}_i) - |b^*(\xb) - b^*(\bm{x}_i)|\\
            &\geq \mathcal{L}_b \cdot r_i - \mathcal{L}_b \cdot r_i = 0
        \end{align*}
        implying that $ \exists \ub \in \mathcal{U}:\bm{n}(\xb)^T \bm{F}(\xb, \ub) \geq 0$. Since, by Thm.~\ref{thm:main}, $\partial \Omega$ is a subset of the union of all balls $\mathcal{B}_{r_i}(\bm{x}_i)$ we know that for any point on $\partial\Omega$ there exists a flow into the set $\Omega$ which by Nagumo's theorem implies that $\Omega$ is control invariant.
    \end{proof}
    To synthesize virtual control inputs for safe expansion of the invariant set, we therefore solve the following quadratic program (QP) with finitely many constraints
    \begin{equation}
        \begin{aligned}
            \bm{\eta}^s(\bm{s}_{\Omega}) = \argmin_{\bm{\eta} \in \mathbb{R}^{n\cdot N}} \quad & \left(\bm{\eta} - \bm{\eta}^{\mathrm{ref}}\right)^T {\bm{Q}} \left(\bm{\eta} - \bm{\eta}^{\mathrm{ref}}\right)\\
            \textrm{s.t.~~} \quad & \frac{\partial \tilde{b}^*(\bm{x}_i)}{\partial \bm{s}_{\Omega}} \bm{\eta} \geq - \alpha(\tilde{b}^*(\bm{x}_i)), i=1,\dots,N
        \end{aligned}
        \label{eq:QPFinal}
    \end{equation}
    where we emphasize that the additional arguments in $\tilde{b}^*(\bm{x}_i) := b^*(\bm{x}_i) - \mathcal{L}(\bm{s}_{\Omega}) \cdot r_i(\bm{s}_{\Omega})$ depend on the state $\bm{s}_{\Omega}$ and have to be taken into account when calculating gradients.
    \begin{corollary}
        Given an initial control invariant set defined by $\bm{s}_{\Omega}(t_0)$ and the dynamics in Eq.~\eqref{eq:PolytopeSystem}, if the QP in Eq.~\eqref{eq:QPFinal} is feasible at all times, then the virtual control input $\bm{\eta}^s$ ensures that $\Omega$ is control invariant $\forall t \geq t_0$.
    \end{corollary}
    \begin{remark}
        Note that, we do not rely on the expansion CBF candidate being a valid CBF since we can always check if the constraints in Eq.~\eqref{eq:QPFinal} hold. In case the QP becomes infeasible, or any $\tilde{b}^*<0$, we can simply use the last obtained solution as a control invariant set.
    \end{remark}
    We use the QP in Eq.~\eqref{eq:QPFinal} in a forward simulation of the closed loop dynamical system~\eqref{eq:PolytopeSystem}, i.e. $\dot{\bm{s}}_{\Omega} = \bm{\eta}^s({\bm{s}}_{\Omega})$. At any point in time, we can use the set state ${\bm{s}}_{\Omega}$ to construct a CBF for the original system~\eqref{eq:deterministic} which we describe in the following section.
    \subsection{CBF Candidates From Control Invariant Set}
    \label{sec:CBFQP}
    So far, we have only considered the boundary of a control invariant set. This boundary could be used to define a switching controller for the system~\eqref{eq:deterministic} that if $\bm{x}(t) \in \partial \Omega$ applies the solution $\ub^*$ of Eq~\eqref{eq:LPfeasibility}. However, this can result in undesired chattering behavior \cite{choi2021robust}. To address this, we design a CBF candidate from the boundary $\partial \Omega$ which results in a smoother control signal. Similar to \cite{singletary2022safety}, we propose to use the signed distance function (SDF) as a CBF candidate for the system \eqref{eq:deterministic}. The signed distance between a state and the boundary $\mathrm{sd}(\xb,\partial\Omega)$
    can be obtained as
    \begin{align*}
    \mathrm{dist}(\xb, \partial\Omega) &:= \min \left\{ \|\bm{d}\|_2 : (\xb + \bm{d}) \in \mathcal{X} \backslash \Omega \right\}, \\
    \mathrm{pen}(\xb, \partial\Omega) &:= \min \left\{ \|\bm{d}\|_2 : (\xb + \bm{d}) \notin \mathcal{X} \backslash \Omega \right\}, \\
    \mathrm{sd}(\xb, \partial\Omega) &:= \mathrm{dist}(\xb, \partial\Omega) - \mathrm{pen}(\xb, \partial\Omega),
    \end{align*}
    which we use as a candidate function \( h(\bm{x}) = \mathrm{sd}(\xb, \partial\Omega) \), taking zero values on the boundary $\partial \Omega$.
    
   It is important to note that we cannot guarantee \(h(\bm{x})\) to be a valid CBF by Def.~\ref{def:CBF} since we only ensure a positive flow into the set $\Omega$ on the boundary. To address this, we introduce a slack variable \(\delta\) into the final CBF-QP formulation 
   \begin{equation}
        \begin{aligned}
        \bm{u}^* = \arg & \min_{\bm{u} \in \mathcal{U}, \, \delta}   (\bm{u} - \bm{u}^\text{ref})^T (\bm{u} - \bm{u}^\text{ref}) + k_s \delta^2  \\
        \text{s.t.} \quad &\frac{\partial h}{\partial \bm{x}} \left( \bm{f}\left(\bm{x}\right) + \bm{g}\left(\bm{x}\right) \bm{u} \right) \geq -\alpha( h(\bm{x})) - \delta, \\
        & \delta \geq 0
        \end{aligned}
    \label{eq:cbf-qp}
    \end{equation}
   rendering the QP always feasible for any reference control input $\ub^{\mathrm{ref}}$ and $k_s > 0$. However, this does not affect the derived control invariance of $\Omega$ since the gradient of $h\ofx$ aligns with the normal $\bm{n}\ofx$ on the zero levelset of $h$. Consequently, it is ensured that there exists a control input according to Eq.~\eqref{eq:existence}. Therefore, the slack variable will always be zero on the boundary.
   This requirement is justified by the fact that \(h(\bm{x})\) is continuously differentiable at the boundary due to the assumption of a $C^1$ manifold. 

\begin{figure}[t!]
    \centering
    \includegraphics[scale=0.7]{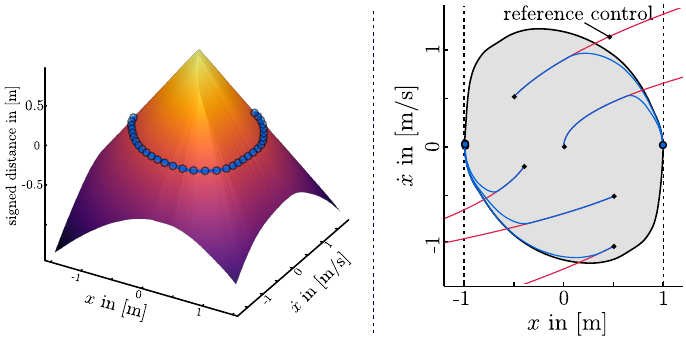} 
    \caption{Simulation results for the running example of a double integrator. \textbf{Left:} SDF constructed numerically from the boundary $\partial \Omega$. Blue dots indicate the position of control points in the phase space. \textbf{Right:} Illustration of the invariant set $\Omega$ with sampled state trajectories shown in blue. If we do not apply the proposed safety filter, the reference controller (red) drives the state into the unsafe area.}
    \label{fig:SDFPlot}
    \vspace{-0.5cm}
    \end{figure}

    \section{Results}
    In this section, we demonstrate how our proposed method can be used to find valid control invariant sets and show how it can be directly used to define a CBF-QP controller. We finalize our running example and also show results for a nonlinear system, namely the inverted pendulum. We implemented our approach in the Julia programming language and use a 4th order Runge-Kutta method to integrate the set dynamics in Eq.~\eqref{eq:PolytopeSystem}.
    
    \subsection{Running Example}
    We define an initial set as a small Ball around the equilibrium state $\xb_{\mathrm{eq}}=[0~0]^T$ since an equilibrium always provides a control invariant set for any nonlinear system. To ensure that the Catmull-Rom curve is fully contained within the union of balls as required in Thm.~\ref{thm:main}, we enforce the constraints on the midpoints of a curve segment and set the radius to $r_i = \tfrac{1}{2}(\tau_2 - \tau_1)$ (see Fig.~\ref{fig:CatmullRom}). Figure~\ref{fig:FirstPage} shows the time evolution of three control invariant sets for $N\in\{10, 20, 50\}$ control points. As one would expect, we can see that the size of the converged set increases as we increase the number of control points. This is due to the radius $r$ decreasing when more Catmull-Rom segments are considered. Since we enforce $b^*(\xb) \geq \mathcal{L}_b(\xb) \cdot r(\xb)$, it is expected that $\Omega$ approaches the maximum invariant set as $N \rightarrow \infty$. Nonetheless, we want to emphasize that all three sets in Fig.~\ref{fig:FirstPage} are control invariant sets, verified on a continuous domain.

    We construct an SDF and obtain a CBF-QP based controller as described in Sec.~\ref{sec:CBFQP}. Figure \ref{fig:SDFPlot} shows the invariant set for $N=50$ control points as well as the constructed SDF. It can clearly be seen that the SDF has non-smooth areas and is, thus, not continuously differentiable. However, its gradients on its zero levelset align with the normals of $\partial \Omega$. For simulation purposes, we sample random initial conditions $\xb_0 \in \Omega$ and reference controls $u^{\mathrm{ref}}$ as shown in Fig.~\ref{fig:SDFPlot}. We can observe that the safe control input synthesized by Eq.~\eqref{eq:cbf-qp} minimally corrects the reference controller while keeping the state within $\Omega$. 
    \begin{figure}[t!]
    \centering
    \includegraphics[scale=1]{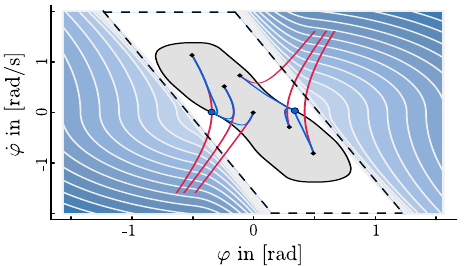} 
    \caption{Simulation results for the inverted pendulum system. The value function calculated using HJ reachability is shown by the blue contours and its zero levelset (maximum control invariant set) is shown by the dashed lines. Our obtained control invariant set is shown in grey. The colored trajectories illustrate forward simulations of the dynamical system.}
    \label{fig:InvPendPlot}
    \vspace{-0.5cm}
    \end{figure}
    \subsection{Inverted Pendulum}
    Next, we find a control invariant set for the inverted pendulum system with its nonlinear dynamics given by
    \begin{align*}
        m \ell^2 \ddot{\varphi} = m g \ell ~\mathrm{sin} (\varphi) + u ,
    \end{align*}
    where $m, \ell$ and $g$ denote the mass, length and gravity, respectively. Again, we assume that we have an actuation constraint on the torque $|u|\leq u_{\mathrm{max}}$ and initialize a small set around the equilibrium state. We define the physical safe set as $\mathcal{S} = \{\varphi, \dot{\varphi} \in \mathbb{R}\mid -\tfrac{\pi}{2} \leq \varphi \leq \tfrac{\pi}{2}, -2 \leq \dot{\varphi}\leq 2\}$ and obtain the maximum control invariant set using HJ reachability on a discretized state space. Figure~\ref{fig:InvPendPlot} shows the contours of the value function as well as its zero level set through the dashed lines. We use $N=50$ control points to find the control invariant set depicted in Fig.~\ref{fig:InvPendPlot}. It can be observed that we do not approach the maximum control invariant which is due to the conservatism induced by Lipschitz constants. Also, we want to highlight that the inverted pendulum system has a theoretically unbounded maximum control invariant set which can currently not be handled by our proposed method. Therefore, we use the set $\mathcal{S}$ with additional constraints on the angular velocity.

    \section{Conclusion and Discussion}
    This paper presents a method for finding a valid control invariant set by expanding the boundary of an initial valid subset. This expansion is formulated as another invariance problem in the space of possible boundaries. As a result, we guarantee control invariance of the expanded set at any point in time which at convergence can be used to obtain a maximized invariant set. The key to verification of the invariant set on a continuous domain is based on leveraging Lipschitz constants of linear programs which allows us to transform the verification problem into a finite number of LPs. Finally, we construct a minimally invasive safety filter in a CBF-QP framework. Numerical experiments demonstrate the effectiveness of the proposed method in both linear and nonlinear system dynamics. 

    Although the theoretical results derived in this work hold for arbitrary state dimensions, the assumption of having a $C^1$ manifold passing through all control points is quite limiting. This is why we have focused on toy examples in 2D where the boundary can be directly parameterized using Catmull-Rom curves. Furthermore, when expanding the invariant set, we have to solve a QP with decision variables of dimension $n\cdot N$, i.e. the number of control points times the state dimension. Although QPs can be solved efficiently using, e.g., recently proposed GPU-accelerated QP solvers \cite{bishop2024relu}, we also suffer from the curse of dimensionality. However, we believe there is plenty of space for future contributions.

    In future work, we aim to extend the proposed approach to higher dimensions either by representing the boundary with manifold splines \cite{gu2005manifold} or orthogonal neural networks \cite{prach2022almost} as these offer known Lipschitz constants. Further interesting directions include modeling the joint dynamical system between the physical state $\bm{x}$ and the set state $\bm{s}_{\Omega}$ which would allow for online expansion of invariant set based on, e.g., sensor observations. Lastly, we believe that Lipschitz constants of LPs can have a broader impact on verification of neural CBFs by accelerating the verification process.
    

    
    

    \bibliographystyle{IEEEtran}
    \bibliography{refs.bib}

\end{document}